\documentclass[11pt]{article}

\usepackage{todonotes}

\usepackage{fullpage}
\usepackage{graphicx}
\usepackage{amssymb}
\usepackage{amsmath}
\usepackage{amsthm}
\usepackage{epsfig}
\usepackage{tikz}
\usetikzlibrary{matrix}
\usepackage{color}
\usepackage{epsfig}
\usepackage{caption}
\usepackage{subcaption}
\usepackage{hyperref}
\newtheorem{theorem}{Theorem}[section]
\newtheorem{lemma}[theorem]{Lemma}

\newtheorem{corollary}[theorem]{Corollary}

\newtheorem{conjecture}[theorem]{Conjecture}

\newcommand{\bb}[1]{\mathbb{#1}}

\begin{document}

\begin{center}

{\Large \bf Large Supports are required for Well-Supported Nash Equilibria}\\ [1cm]

{\sc Yogesh Anbalagan}\footnote[1]{School of Computer Science, McGill University. Email: {\tt yogesh.anbalagan@mail.mcgill.ca}},
{\sc Hao Huang}\footnote[2]{Institute for Mathematics and its Applications, University of Minnesota. Email: {\tt huanghao@ima.umn.edu}},
{\sc Shachar Lovett}\footnote[3]{Computer Science \& Engineering Department, University of California, San Diego. Email: {\tt slovett@cs.ucsd.edu}},\\
{\sc Sergey Norin}\footnote[4]{Department of Mathematics and Statistics, McGill University. Email: {\tt snorin@math.mcgill.ca}},
{\sc Adrian Vetta}\footnote[5]{Department of Mathematics and Statistics,
and School of Computer Science, McGill University.
Email: {\tt vetta@math.mcgill.ca}},
and 
{\sc Hehui Wu}\footnote[6]{Department of Mathematics,  University of Mississippi.
Email: {\tt hhwu@olemiss.edu}}
 \\[.5cm]
 
\end{center}

\begin{abstract}
We prove that for any constant $k$ and any $\epsilon<1$, there exist bimatrix win-lose games for which every $\epsilon$-WSNE
requires supports of cardinality greater than $k$.
To do this, we provide a graph-theoretic characterization of win-lose games
that possess $\epsilon$-WSNE with constant cardinality supports. We then apply a 
result in additive number theory of Haight \cite{Hai73}
to construct win-lose games that do not satisfy the requirements of the characterization.
These constructions disprove graph theoretic conjectures of 
Daskalakis, Mehta and Papadimitriou~\cite{DMP09} and Myers~\cite{Mye03}.
\end{abstract}

\section{Introduction}
A Nash equilibrium of a bimatrix game $(A,B)$ is a pair of strategies that are mutual best-responses.
Nash equilibria always exist in a finite game~\cite{N51}, but 
finding one is hard, unless $PPAD \subseteq P$~\cite{CDT09}. 
This has lead to the study of relaxations of the equilbrium concept.
A notable example is an \emph{$\epsilon$-approximate Nash equilibrium} ($\epsilon$-NE).  
Here, every player must receive an expected payoff within $\epsilon$ of their best response payoff. 
Thus $\epsilon$-NE are numerical relaxations of Nash Equilibria.
Counterintuitively, however, given that Nash's existence result is via a fixed point theorem,
Nash equilibria are intrinsically combinatorial objects. In particular, the crux of the equilibrium problem
is to find the supports of the equilibrium.
In particular, at an equilibrium, the supports of both strategies consist only of best responses.
This induces a combinatorial relaxation called an \emph{$\epsilon$-well supported approximate Nash 
equilibrium} ($\epsilon$-WSNE). Now the content of the supports are restricted, but less stringently than in
an exact Nash equilibrium.
Specifically, both players can only place positive probability on strategies that have payoff
within $\epsilon$ of a pure best response.

Observe that in an $\epsilon$-NE, no restriction is placed on the supports of the strategies. 
Consequently, a player might place probability on a strategy that is arbitrarily far from being a best response!
This practical deficiency is forbidden under $\epsilon$-WSNE.
Moreover, the inherent combinatorial structure of $\epsilon$-WSNE has been extremely useful
in examining the hardness of finding Nash equilibria.
Indeed, Daskalakis, Goldberg and Papadimitriou~\cite{DGP09} introduced $\epsilon$-WSNE in proving 
the PPAD-completeness of finding a Nash equilibrium in multiplayer games. 
They were subsequently used as the notion of approximate equilibrium by Chen, Deng and Teng~\cite{CDT09} 
when examining the hardness of bimatrix games.

This paper studies the (non)-existence of $\epsilon$-WSNE with small supports.
Without loss of generality, we may assume that all payoffs in $(A,B)$ are in $[0,1]$. 
Interestingly, for $\epsilon$-NE, there is then a simple $\frac12$-NE with supports
of cardinality at most two ~\cite{DMP09}. Take a row $r$. Let column $c$ be a best response to $r$, and
let $r'$ be a best response to $c$. 
Suppose the row player places probability $\frac12$ on $r$ and $r'$, and 
the column player plays column $c$ as a pure strategy. It is easy to verify that this
is a $\frac12$-NE. On the other hand, Alth\"ofer \cite{Alt94} showed the existence of zero-sum games for 
which every $\epsilon$-NE, with $\epsilon<\frac14$, requires supports of cardinality at least $\log n$. 
This result is almost tight; a probabilistic argument shows the existence of
$\epsilon$-NE with supports of cardinality $O(\frac{log \, n}{\epsilon^2})$, for any $\epsilon > 0$; see Lipton et al.~\cite{LMM03}. 

For the case of  well-supported equilibria, Anbalagan et al. \cite{ANS13} recently showed the existence of win-lose games 
for which every $\epsilon$-WSNE, with $\epsilon<\frac23$, require supports of cardinality at least $\sqrt[3]{\log n}$. 
They also proved, in contrast to $\epsilon$-NE, that with supports of cardinality at most two, it is not possible to guarantee
the existence of an $\epsilon$-WSNE, for any $\epsilon<1$.


The outstanding open problem in the area is whether there is a constant $k$ and an $\epsilon<1$ 
such that, for any bimatrix game, there is a $\epsilon$-WSNE
with supports of cardinality at most $k$. In the paper we prove this is not the case.
This result illustrates a fundamental structural distinction between
$\epsilon$-WSNE and $\epsilon$-NE. This structural distinction also has practical implications 
with regards to behavioural models and popular equilibria search algorithms that focus upon small 
supports. The key to our result is the disproof of graph theoretic conjectures of 
Daskalakis, Mehta and Papadimitriou~\cite{DMP09} and Myers~\cite{Mye03}
via an old result in additive number theory of
Haight~\cite{Hai73}.

\section{WSNE and a Graph Theoretic Conjecture}

A bimatrix game is a 2-player game with $m \times n$ payoff matrices $A$ and $B$. We consider normal form games 
with entries in the payoff matrices in $[0,1]$. A pair of mixed strategies $\{p,q\}$ forms an {\em$\epsilon$-well supported
Nash equilibrium} ({\em $\epsilon$-WSNE)} if every pure strategy in the support of $p$ (resp. $q$) is 
an $\epsilon$-approximate best response to $q$ (resp $p$). Thus $\{p,q\}$ forms an $\epsilon$-WSNE if and only if:
$$ \forall i: {p}_i > 0 \: \Rightarrow \quad {e_i}^TAq \geq {e_j}^TAq - \epsilon  \quad \forall j=1,..,m $$
and
$$ \forall i: {q}_i > 0 \: \Rightarrow \quad {p}^TBe_i \geq {p}^TBe_j - \epsilon \quad \forall j=1,..,n $$

To analyse well-supported equilibria in a win-lose game $(A,B)$,  Daskalakis et al. \cite{DMP09}
applied a {\em decorrelation transformation} to obtain a pair of decorrelated matrices $(A^*,B^*)$.
The exact details of this decorrelation transformation are not important here.
What is pertinent, however, is that the $n\times n$ square $0-1$ matrix $A^*$ induces a directed, possibly non-bipartite,
graph $H=(V,E)$. There are $n$ vertices in $V$, and there is an arc $ij\in E$ if an only if $A^*_{ij}=1$. 
Moreover, Daskalakis et al. proved that the original win-lose game has
a $(1-\frac{1}{k})$-WSNE with supports of cardinality at most $k$ if $H$
contains either a directed cycle of length $k$ or a set of $k$ undominated\footnote{A set $S$ is \emph{undominated} if 
there is no vertex $v$ that has an arc to every vertex in $S$.}  vertices.
Furthermore, they conjectured that every directed graph contains either a small cycle
or a small undominated set.

\begin{conjecture}\label{conj:kl}\cite{DMP09}
There are integers $k$ and $l$ such that every digraph either has a cycle of length at most $k$ or an undominated set of $l$ vertices. 
\end{conjecture}
Indeed, they believed the conjecture was true for $k=l=3$ and, consequently, that every
bimatrix win-lose game has a $\frac23$-WSNE with supports of cardinality at most three. 
Interestingly, motivated by the classical Caccetta-Haggkvist conjecture \cite{CH78} 
in extremal graph theory, 
a similar conjecture was
made previously by Myers \cite{Mye03}.
\begin{conjecture}\label{conj:myers}\cite{Mye03}
There is an integer $k$ such that every digraph either has a cycle of length at most $k$ or an undominated set of two vertices. 
\end{conjecture}
Myers conjectured that this was true even for $k=3$, but Charbit \cite{Cha05} proved this special case to be false.

We say that $D$ is a \emph{$(k,l)$-digraph} if every directed cycle in $D$ has length at least $k$, and 
every $S \subseteq V(D)$ of cardinality at most $l$ is 
dominated. In Section \ref{sec:counterex}, we will prove that there exists a finite $(k,l)$-digraph
for every pair of positive integers $k$ and $l$. This will imply that Conjectures \ref{conj:kl} and \ref{conj:myers} 
are false.



\section{A Characterization for Games with Small Support $\epsilon$-WSNE.}

In this section, we show Daskalakis et al.'s sufficiency condition extends to a 
characterization of when a win-lose game has $\epsilon$-WSNE with constant supports. 
To do this, rather than non-bipartite graphs, it is more natural for
bimatrix games to work with bipartite graphs.
In particular, any win-lose game $(A,B)$ has a simple representation
as a bipartite directed graph $G=(R \cup C,E)$. 
To see this, let $G$ contain a vertex for each row and a vertex for each column. There exists an arc $(r_i,c_j)\in E$ if and 
only if $A_{ij}=1$. So $r_i$ is the best response for the row player against the strategy $c_j$
of the column player. Similarly, there exists an arc $(c_j,r_i) \in E$ if and only if $B_{ij}=1$. So, $c_j$ is a best 
response for the column player against the strategy $r_i$ of the row player. 

We will now show that a win-lose game has a $(1-\frac{1}{k})$-WSNE with supports of cardinality at most $k$ 
{\em if and only if} the corresponding directed bipartite graph has either a small cycle or a small set of undominated 
vertices. Thus we obtain a characterization of win-lose games that have $\epsilon$-WSNE with small 
cardinality supports.

It what follows, we will only consider undominated sets that are contained either in $R$ or in $C$
\begin{lemma}\label{lem:undominated6}
Let $G$ be a win-lose game with minimum out-degree at least one. 
If $G$ contains an undominated set of cardinality $k$ then there is a 
$(1-\frac{1}{k})$-WSNE with supports of cardinality at most~$k$.
\end{lemma}
\begin{proof}
Without loss of generality, let $U=\{r_1,...,r_k\}$ be the undominated set. 
Let the row player play a uniform strategy $p$ on these $k$ rows. Since $U$ is undominated, any
column has expected payoff at most $1-\frac{1}{k}$ against $p$.
Therefore every column $c_j$ is a $(1-\frac{1}{k})$-approximate best response against $p$.

By assumption, each row vertex $r_i$ has out-degree at least one.
Let $c_{f(i)}$ be an out-neighbour of $r_i$ (possibly $f(i)=f(j)$ for $j\neq i$). 
Now let the column player play a uniform strategy $q$ on $\{c_{f(i)}\}_{i=1}^k$. Because 
$q$ has support cardinality at most $k$, each pure strategy $r_i\in U$ has an expected
payoff at least $\frac{1}{k}$ against $q$.
Thus, these $r_i$'s are all $(1-\frac{1}{k})$-approximate best responses for the row player against $q$.
So $\{p,q\}$ is a $(1-\frac{1}{k})$-WSNE with supports of cardinality at most $k$.
\end{proof}

\begin{lemma}\label{lem:cycle6}
If $G$ contains a cycle of length $2k$ then there is
a $(1-\frac{1}{k})$-WSNE with supports of cardinality $k$.
\end{lemma}
\begin{proof}
Let $W$ be a cycle of length $2k$ in $G$. Since $G$ is bipartite, $k$ of the vertices in the cycle are 
row vertices and $k$ are column vertices. Let $p$ be the uniform strategy on the rows in $W$ 
and let $q$ be the uniform strategy on the columns in $W$. We claim that $p$ and $q$ form
a $(1-\frac{1}{k})$-WSNE. To prove this, consider the subgraph $F$ induced by the vertices of $W$. 
Every vertex in $F$ has out-degree (and in-degree) at least one since $W\subseteq F$. So, every pure strategy in $p$, 
gives the row player an expected payoff of at least $\frac{1}{k}$ against $q$. Thus, every pure strategy in $p$ is a 
$(1-\frac{1}{k})$-best response for the row player against $q$. 
Similarly, every pure strategy in $q$ is a $(1-\frac{1}{k})$-best response for the column player against $p$. 
\end{proof}

Lemma \ref{lem:undominated6} and Lemma \ref{lem:cycle6} immediately give the following corollary.
\begin{corollary}\label{cor:main6}
Let $G$ be a win-lose game with minimum out-degree at least one. 
If $G$ contains a cycle of length $2k$ or an undominated set of cardinality $k$ 
then then the win-lose game has $(1-\frac{1}{k})$-WSNE with supports of cardinality at most $k$. \qed
\end{corollary}

Importantly, the converse also holds. 
\begin{lemma}\label{lem:mainconverse}
Let $G$ be a win-lose game with minimum out-degree at least one. 
If there is an $\epsilon$-WSNE (for any $\epsilon <1$) with supports of cardinality at most $k$ 
then $G$ either contains an undominated set of cardinality $k$ or contains 
a cycle of length at most $2k$. 
\end{lemma}
\begin{proof}
Take a win-lose game $G=(R \cup C,E)$ and let $p$ and $q$ be an $\epsilon$-WSNE.
Suppose the supports of $p$ and $q$, namely $P\subseteq R$ and $Q\subseteq C$, have 
cardinality at most $k$. 

We may assume that every set of cardinality every set of $k$ (on the same side of the bipartition) is dominated; otherwise 
we are already done. In particular, both $P$ and $Q$ are dominated.
Consequently, the row player has a best response with expected payoff $1$ against
$q$. Similarly, the column player has a best response with expected payoff $1$ against
$p$. Thus, for the $\epsilon$-WSNE $\{p, q\}$, we have:
\begin{eqnarray*}
\forall i: \ \ {p}_i > 0 &\Rightarrow& \quad {e_i}^TRq \geq 1 - \epsilon >0 \\
\forall j: \ \ {q}_j > 0 &\Rightarrow& \quad {p}^TCe_j \geq 1- \epsilon >0 
\end{eqnarray*}
Here the strict inequalities follow because $\epsilon <1$.
Therefore, in the subgraph $F$ induced by $P\cup Q$, every vertex 
has an out-degree at least one. But then $F$ contains a cycle $W$. 
Since $H$ contains at most $2k$ vertices, the cycle $W$ has length at most $2k$.
\end{proof}

Corollary \ref{cor:main6} and Lemma \ref{lem:mainconverse} then give the following characterization for win-lose games 
with $\epsilon$-WSNE with small cardinality supports

\begin{theorem}\label{thm:char}
Let $G$ be a win-lose game with minimum out-degree at least one. 
Take any constant $k$ and any $\epsilon$ such that $1-\frac{1}{k} \leq \epsilon <1$.
The game contains an $\epsilon$-WSNE with supports of cardinality at most $k$ 
if and only if $G$ contains 
an undominated set of cardinality $k$ or  a cycle of length at most $2k$. 
\qed
\end{theorem}

\section{Digraphs of Large Girth with every Small Subset Dominated}\label{sec:counterex}
In this section, we will first prove that there exists a finite $(k,l)$-digraph
for every pair of positive integers $k$ and $l$ and, hence, disprove Conjecture \ref{conj:kl}.
We then adapt the resulting counterexamples in order to apply Theorem \ref{thm:char}
and deduce that, for any constant $k$ and any $\epsilon<1$, there exist bimatrix win-lose games 
for which every $\epsilon$-WSNE require supports of cardinality greater than $k$.

The main tool we require is a result of Haight~\cite{Hai73} from additive number theory. 
We will require the following notation. Let $\Gamma$ be an additive group. Then, for $X \subseteq \Gamma$, denote 
\begin{eqnarray*} 
X-X &=& \{x_1 -x_2 \: | \: x_1,x_2 \in X\},\: \mathrm{and} \\
(k)X &=& \{x_1+x_2+\ldots+x_k \: | \: x_i \in X \:\mathrm{for} \: 1 \leq i \leq k\}.
\end{eqnarray*}
Finally, let $\bb{Z}_q=\{0,1,\dots, q-1\}$ denote the additive group of $\bb{Z}/q\bb{Z}$, the integers modulo $q$.
Haight~\cite{Hai73} proved:
\begin{theorem}\cite{Hai73}\label{thm:Haight}
For all positive integers $k$ and $l$, there exists a positive integer $\hat{q}=\hat{q}(k,l)$ and a 
set $X \subseteq  \bb{Z}_{\hat{q}}$, such that $X-X = \bb{Z}_{\hat{q}}$, 
but $(k)X$ omits $l$ consecutive residues. \qed
\end{theorem}
To construct the finite $(k,l)$-digraph we will use the following corollary.
\begin{corollary}\label{cor:avoidZero}
For every positive integer $k$, there exists a positive integer $q=q(k)$ and a 
set $Y \subseteq  \bb{Z}_q$, such that $Y-Y = \bb{Z}_q$, 
but $0 \not \in (k)Y$.
\end{corollary}
\begin{proof}
Let $l=k$ and apply Theorem~\ref{thm:Haight} with $q=q(k)=\hat{q}(k,k)$. 
Thus, we obtain a set $X \subseteq \bb{Z}_q$ with the properties that: (i) $X-X = \bb{Z}_q$,
and (ii) $(k)X$ omits $k$ consecutive residues. 
But these $k$ consecutive residues must contain $ky$ for some $y \in \bb{Z}_q$.
Thus, there exists $y \in \bb{Z}_q$ such that $ky \not \in (k)X$. 

Now, define $Y:=X-y$. Then $Y-Y = X-X =  \bb{Z}_q$.
Furthermore, $ky \not \in (k) (Y+y)$. This implies that $0 \not \in (k)Y$, as desired.
\end{proof}

We now construct a counter-example to Conjecture \ref{conj:myers} of Myers. We will then show how
the construction can be extend to disprove Conjecture \ref{conj:kl}.
\begin{theorem}\label{thm:k2digraph}
For any positive integer $\kappa$, there exists a $(\kappa,2)$-digraph $D$.
\end{theorem}
\begin{proof}
Set $k= (\kappa-1)!$ and apply Corollary~\ref{cor:avoidZero}.
Thus we find $Y \subseteq  \bb{Z}_q$, with $q=q(k)$ where $Y-Y = \bb{Z}_q$, 
and $0 \not \in (k)Y$.
From $Y$, we create a directed graph $D$ as follows.
Let the vertex set be $V(D)=\bb{Z}_q$. Let the arc set be 
 $E(D)=\{z_1z_2 \: | \: z_1 -z_2 \in Y\}$.

Now take any pair $z_1,z_2 \in \bb{Z}_q$.
Because $Y-Y = \bb{Z}_q$, there exist $y_1,y_2 \in Y$ such that $z_1-z_2=y_1-y_2$.
We now claim that the vertex pair $z_1, z_2\in V(D)$ is dominated. 
To see this consider the vertex $x\in V(D)$ where $x=z_1+y_2=z_2+y_1$.
Then $xz_1$  is an arc in $E(D)$ because $x-z_1=(z_1+y_2)-z_1=y_2\in Y$.
On the other hand $x=z_2+y_1$ and so $x-z_2=y_1\in Y$.
Consequently, $xz_2$ is also in $E(D)$.
Hence, every subset of $V(D)$ of cardinality at most $2$ is dominated.

It remains to prove that $D$ contains no directed cycle of length less than $\kappa$.
So, assume there is a cycle $C$ with ordered vertices $z_1,z_2,\ldots, z_s$, where $s<\kappa$.
As $z_iz_{i+1}$ is an arc we have that $z_i-z_{i+1}=y_i$ where $y_i\in Y$, for $1\leq i \leq s$
(here we assume $z_{s+1}=z_1$). Summing around the cycle we have that $y_1+y_2+\cdots +y_s=0$ modulo $q$.
This implies that $0 \in (s)Y$ as $y_1,y_2,\ldots,y_s \in Y$.
Consequently, $0 \in (ts)Y$ for any positive integer $t$. In particular, $0 \in (k)Y=((\kappa-1)!)Y$, as $s\le \kappa-1$.
This contradicts the choice of $Y$ and, so, $D$ is a $(\kappa, 2)$-digraph, as desired.    
\end{proof}

\begin{theorem}\label{thm:klexist} 
For every pair of positive integers $k$ and $l$, there exists a finite $(k,l)$-digraph.
\end{theorem}
\begin{proof}
Without loss of generality, assume $l \geq 2$.
By Theorem~\ref{thm:k2digraph}, there exists a $((k-1)(l-1)+1,2)$-digraph $D'$.
We claim that the $(l-1)$-st power of $D'$ is a $(k,l)$-digraph. 
More precisely, let the digraph $D$ be defined by
$V(D)=V(D')$ and $vw \in E(D)$ if and only if there exists a directed walk from $v$ to $w$ in $D'$ using at most $(l-1)$ edges.

Suppose $D$ has a cycle of length at most $k-1$. This corresponds to a closed directed walk of length a most $(k-1)(l-1)$ in $D'$.
This is a contradiction as $D'$ has no cycles of length shorter than $(k-1)(l-1)+1$. 
Therefore, the shortest directed cycle in $D$ has length at least $k$.

It remains to prove that every $S \subseteq V(D)$ with $|S|=l$ is dominated. 
So take $S=\{v_1,v_2,\ldots,v_l\}$. 
Recall that every pair of vertices in $V(D')=V(D)$ is dominated in $D'$.
So there is a vertex $z_1$ dominating $v_1$ and $v_2$ in $D'$.
Now let $z_{i+1}$ be a vertex dominating $z_i$ and $v_{i+2}$ for $1 \leq i \leq l-2$. 
By construction, there is a directed walk in $D'$ from $z_{l-1}$ to $v_i$ of length at most $l-1$, for every $1 \leq i \leq l$. 
Thus $z_{l-1}v_i\in E(D)$, and $S$ is dominated in $D$, as desired.
\end{proof}

Observe that these constructions are non-bipartite. To exploit the characterization of Theorem~\ref{thm:char}
(and therefore conclude that there are games with no $\epsilon$-WSNE with small supports), 
we desire bipartite constructions. These we can create using a simple mapping from non-bipartite to bipartite graphs.
Given a non-bipartite graph $G=(V, E)$, we build a win-lose game,
that is, a bipartite directed graph $G'=(R\cup C,E')$ as follows.
We set $R=C=V$. Thus, for each $v_i\in V$ we have a row vertex $r_i\in R$ and a column vertex $c_i\in C$.
Next, for each arc $a=(v_i,v_j)$ in $G$, we create two arcs $(r_i,c_j)$ and $(c_i,r_j)$ in $G'$. Finally,
for each $v_i\in V$ we add an arc $(r_i ,c_i)$.

Now let's understand what this mapping does to cycles and undominated sets.
First, suppose $G$ contains a cycle of length $k$. Then observe that $G'$ contains a cycle of length $k$ if $k$ is even
and of length $k+1$ if $k$ is odd. On the other hand, suppose the minimum length cycle in $G'$ is $k+1$.
This cycle will contain at most one pair of vertices type $\{r_i ,c_i\}$, and if it contains such a pair then these vertices are 
consecutive on the cycle. (Otherwise we can find a shorter cycle in $G'$.) Thus,
$G$ contains a cycle of length $k$ or $k+1$.
 
Second, consider an undominated set $S\subseteq V$ of size $\ell$ in $G$. Then $S\subseteq R$ is
undominated in $G'$. (Note $S\subseteq C$ may be dominated because we added arcs of the form $(r_i, c_i)$ to $G'$). 
On the other hand if $S$ is undominated in $G'$ (either in $R$ or $C$) then $S$ is also undominated in $G$.


Applying this mapping to a non-bipartite $(2k+1,k)$-digraph produces
a bipartite digraph for which every set of $k$ vertices (on the same side of the bipartition) 
is dominated but that has no cycle of length at most $2k$.
Thus, by Theorem \ref{thm:char}, the corresponding game has no $\epsilon$-WSNE, for any $\epsilon<1$,
with supports of cardinality at most $k$.

\begin{theorem}
For any constant $k$ and any $\epsilon<1$, there exist bimatrix win-lose games for which 
every $\epsilon$-WSNE requires supports of cardinality greater than $k$.
\qed
\end{theorem}


\begin{thebibliography}{99}

\bibitem{Alt94} I. Alth\"ofer, ``On sparse approximations to randomized strategies and convex combinations", {\em Linear Algebra and
its Applications}, {\bf 199}, pp339-355, 1994.

\bibitem{ANS13} Y. Anbalagan, S. Norin, R. Savani and A. Vetta, ``Polylogarithmic supports are required for 
approximate well-supported Nash equilibria below 2/3", {\em Proceedings of Ninth Conference on Web and Internet Economics (WINE)}, pp15-23, 2013.

\bibitem{CH78} L. Caccetta and R. H\"{a}ggkvist, ``On minimal digraphs with given girth", {\em Congressus Numerantium}, {\bf 21}, pp181-187, 1978.

\bibitem{Cha05} P. Charbit, ``Circuits in graphs and digraphs via embeddings", 
{\em Doctoral Thesis}, University of Lyon, 2005.

\bibitem{CDT09} 
X. Chen, X. Deng, and S. Teng, ``Settling the complexity of computing two-player
Nash equilibria'', {\em Journal of the ACM}, {\bf 56(3)}, pp1-57, 2009.

\bibitem{DGP09} C. Daskalakis, P. Goldberg, and C. Papadimitriou, ``The complexity of computing a Nash equilibrium", 
{\em SIAM Journal on Computing}, {\bf 39(1)}, pp195-259, 2009.

\bibitem{DMP09} C. Daskalakis, A. Mehta, and C. Papadimitriou, ``A note on approximate Nash equilibria", 
{\em Theoretical Computer Science}, {\bf 410(17)}, pp1581-1588, 2009.

\bibitem{Hai73} J. Haight, ``Difference covers which have small $k$-sums for any $k$", 
{\em Mathematika}, {\bf 20}, pp109-118, 1973.

\bibitem{FGS12} J. Fearnley, P. Goldberg, R. Savani, and T. S\o rensen, ``Approximate well-supported Nash 
equilibria below two-thirds", {\em Proceedings of Fifth International Symposium on Algorithmic Game Theory (SAGT)}, pp108-119, 2012.

\bibitem{KS10} S. Kontogiannis and P. Spirakis, ``Well supported approximate equilibria in bimatrix games", {\em Algorithmica},
{\bf 57}, pp653-667, 2010.

\bibitem{LMM03} R. Lipton, E. Markakis, and A. Mehta, ``Playing large games using simple startegies", 
{\em Proceedings of Fourth Conference on Electronic Commerce (EC)}, pp36-41, 2003.

\bibitem{LY94} R. Lipton and N. Young, ``Simple strategies for zero-sum games with applications to complexity theory", 
{\em Proceedings of Twenty-Sixth Symposium on Theory of Computing (STOC)}, pp734-740, 1994. 

\bibitem{Mye03} J. Myers, ``Extremal theory of graph minors and directed graphs", 
{\em Doctoral Thesis}, University of Cambridge, 2003.

\bibitem{N51} J. Nash, ``Non-cooperative games", {\em Annals of Mathematics}, Vol 54, 289-295, 1951.
\end{thebibliography}
\end{document}